\documentclass[a4paper,12pt]{amsart}
\usepackage{amsfonts,amsmath,amsthm}
\usepackage{amsmath,amssymb}  % Better maths support & more symbols
\usepackage{graphicx}
\newcommand*{\scrN}{\mathord{\mathcal{N}}}%
\newcommand*{\AAA}{{\mathbb{A}}}% quantum real numbers
\newcommand*{\RR}{{\mathbb{R}}}% real numbers
\newcommand*{\CC}{{\mathbb{C}}}% complex numbers 
\newcommand*{\scrP}{\mathord{\mathcal{P}}}% projection
\newcommand*{\scrS}{\mathord{\mathcal{S}}}% Schwartz space
\DeclareRobustCommand\openone{\leavevmode\hbox{\small1\normalsize\kern-
.33em1}}%

\newcommand*{\EsubS}{\ensuremath{\mathord{\mathcal{E}_{\scrS}}}}
\newtheorem{theorem}{Theorem}
\newtheorem{definition}{Definition}
\newtheorem{lemma}[definition]{Lemma}
\newtheorem{proposition}[definition]{Proposition}

\begin{document}
% You should use BibTeX and apsrev.bst for references
\bibliographystyle{apsrev}
%\preprint{HEP/123-qed}

\title{\textbf{ Quantum real numbers and measurement}\\}
% Force line breaks with \\

\author{J.V. CORBETT}
%\thanks{also at Physics Department, XYZ University.}
%Lines break automatically or can be forced with \\
%\author}%
\email{john.corbett@mq.edu.au}
\address {Department of Mathematics, Macquarie University, N.S.W. 2109} %

\date{\today}
% It is always today, today, but you may specify any date with \date.

\begin{abstract}
The quantum mechanical measurement problem does not arise in the quantum real number approach to quantum measurements of the first kind.  The attributes of individual microscopic systems in the experimental  ensemble always have qr-number values so the individual systems can be followed throughout the process. The interaction with an apparatus connects the qr-number value of the quantity to be measured with the qr-number value of an attribute of the apparatus that can be locally approximated by a classical number and subsequently amplified to a recordable output. 
\end{abstract}

%Valid PACS numbers may be entered using the \verb+\pacs{#1}+ command.
%\end{abstract}

%\pacs{03.65, 03.67}
% PACS, the Physics and Astronomy Classification Scheme.

\maketitle

\section{Introduction}\label{intro} There are broadly two processes in which measurements are used in modern applications of quantum mechanics: the first is to determine the numerical value of a physical attribute of a quantum system, the second is to determine the state of the system\cite{1}.  In this paper the first problem is emphasised.

The measurement problem arises in standard quantum theories of both, for a recent discussion see Schlosshauser\cite{2}. The measurement problem has two parts: \begin{itemize}
\item The problem of definite outcomes.
\item The problem of the preferred basis.
\end{itemize} 
The first occurs because the measurement of a microscopic system $\mathcal{S}$ yields a probability distribution of the values of one or many attributes of $\mathcal{S}$. Any prediction can only be verified by experimental data obtained from an ensemble of identically prepared replicas of $\mathcal{S}$.  In order that the relative frequencies of the various outcomes can be determined, the final ensemble must be such that each outcome is observationally distinct. This basic requirement for determining probabilities is not satisfied in the standard Hilbert space quantum theories.

The second doesn't arise in the qr-number approach because it doesn't accept the premise that a wave function provides a complete state of a quantum system. A complete state in the qr-number model is given by an open set of quantum states, see \S \ref{basics} and \cite{19}.
\subsubsection{The standard description} \label{stan}The following is a simple example, from \cite{1} pp 75-78, illustrating the problem of definite outcomes in the standard quantum mechanical description of a measurement. There are two quantum systems: $\mathcal{S}$, the carrier of an attribute, represented by the operator $\hat A_{S},$ which is to be measured, and $\mathcal{M}$ a measurement apparatus with a pointer represented by an operator $\hat B_{M}$. At time $t_{1},$ $\mathcal{S}$ is in  a superposition of $\hat A_{S}$'s eigenstates $\psi_{S}^{\pm}$ (eigenvalues $\lambda^{\pm}$) while $\mathcal{M}$ is in the null eigenstate $\phi_{M}^{0}$ of its pointer operator $\hat B_{M}$. \begin{equation}\label{EQ1}
\Psi_{S,M}(t_{1}) = \psi_{S} \otimes \phi_{M}^{0} = (\alpha_{+} \psi_{S}^{+} + \alpha_{-} \psi_{S}^{-}) \otimes \phi_{M}^{0}. 
\end{equation} $\psi_{S}^{+} \perp \psi_{S}^{-}$ and $|\alpha_{+}|^{2} + |\alpha_{-}|^{2} = 1$ and all wave functions are unit vectors. The aim is to determine the distribution of the values $\lambda^{\pm}$ of the attribute $\hat A_{S}$ in the prepared system $S$. To this end the relative frequencies, $|\alpha_{+}|^{2}$ and $|\alpha_{-}|^{2}$ of the outcomes $\lambda^{\pm}$,  are determined from the ensemble of prepared systems.

An interaction between $\mathcal{S}$ and $\mathcal{M}$ produces an entangled state,  
\begin{equation}\label{EQ2}
\Psi_{S,M}(t_{2}) = \alpha_{+} \psi_{S}^{+}\otimes \phi_{M}^{+} + \alpha_{-} \psi_{S}^{-} \otimes \phi_{M}^{-},
\end{equation} at the time $t_{2} > t_{1}$. The vectors $\phi_{M}^{+} \perp \phi_{M}^{-}$ and are assumed to be macroscopically distinguishable eigenstates of $\hat B_{M}$.

Because the wave-function $\Psi_{S,M}(t_{2})$ is an entangled pure state and not a mixed state, it is not possible to ascribe a particular wave-function to $\mathcal{M}$. That is, there is no definite outcome associated with the measurement process. If we assume that a pure state describes the state of an ensemble of identical systems, the wave-function $\Psi_{S,M}(t_{2})$ describes that of an ensemble of combined $\mathcal{S}$ and $\mathcal{M}$ systems. It does not parametrise a variety of outcomes and hence does not determine the probabilities of different outcomes. This is the measurement problem.

The pure state $\Psi_{S,M}(t_{2})$ cannot evolve unitarily to a mixed state so that the Schr\"odinger evolution cannot deliver a definite outcome for the measurement problem. This is where the "collapse hypothesis" or "projection postulate" is inserted, \cite {1} \S 2.3.3, pp 86-91.  The standard unitary time evolution of quantum mechanics is replaced by a jump from the pure state, $\rho = \hat P _{\Psi_{S,M}(t_{2})}$, to the mixed state \begin{equation}\label{EQ3}
\rho' = |\alpha_{+}|^{2} \hat P_{\psi_{S}^{+}}\otimes \hat P _{\phi_{M}^{+}} +  |\alpha_{-}|^{2} \hat P_{\psi_{S}^{-}}\otimes \hat P _{\phi_{M}^{-}}
\end{equation} Then the mixed state collapses to $\hat P_{\psi_{S}^{+}}\otimes \hat P _{\phi_{M}^{+}}$ with probability $|\alpha_{+}|^{2}$ or to $\hat P_{\psi_{S}^{-}}\otimes \hat P _{\phi_{M}^{-}}$ with probability $|\alpha_{-}|^{2}.$ The justification of these assumptions is decidedly ad hoc and this has always been a contentious area of quantum mechanics and one which has been often taken as a sign of the incompleteness of the theory. Home \cite {1}, Chapter 2, has a good discussion of the issues that have arisen.

\section{The qr-number description} The qr-number values of the pertinent quantities always have a qr-number values, see \S \ref{basics},  so their trajectories throughout the experiment can be followed. We assume that both the system $S$, whose properties are to be measured and the measurement system $M$ are particles with non-zero masses, $m_{S}$ and $m_{M}$.

In the preparation stage, see \S \ref{prep}, the quantity $\hat Q_{S}$ to be measured is prepared so that it can be measured. An epistemic condition is prepared for an ensemble of $S$-particles. In the generic example $\hat Q_{S}$ has only two eigenvalues, $\lambda_{\pm}$ with eigenvectors $\phi_{\lambda_{\pm}}.$ Let $\psi_{S}(\vec \alpha)= \alpha_{+} \phi_{\lambda_{+}} + \alpha_{-} \phi_{\lambda_{-}}$, with $\vec \alpha= (\alpha_{+},  \alpha_{-}) \in \CC^{2}$ and $| \alpha_{+}|^{2} + | \alpha_{-}|^{2} = 1,$ be the wave-function for $S$ that was used in the standard description, \S \ref{stan}. Then lemma 2 of \S \ref{prep} shows that the prepared condition  is \begin{equation}
W_{S}(\vec \alpha)  = |\alpha_{+}|^{2} W_{S}^{+} +  |\alpha_{-}|^{2} W_{S}^{-}.\end{equation} where $W_{S}^{\pm} =  \mathcal{N}(\hat P_{\phi_{\lambda_{\pm}}},\hat Q_{S}, \delta)$ are designated epistemic conditions on which the qr-number values $q_{S}|_{W_{S}^{\pm}}$ are measurable with $
 q_{S}|_{W_{S}^{\pm} } \approx \lambda_{\pm}.$

For example, when $\hat Q_{S}$ is $\epsilon$-sharp collimated in an interval $I_{\pm}$ centred at the eigenvalue $\lambda_{\pm}$ on the condition $W_{S}^{\pm}$, see definition 1, \S\ref{measreq}, then $\hat Q_{S}$'s qr-number value $q_{S}|_{W_{S}^{\pm}}$ is well approximated  by the eigenvalue $\lambda_{\pm}$. The coefficient $ |\alpha_{\pm}|^{2}$ are the  relative frequencies of preparing $W_{S}^{\pm}$. At the same time the measurement device is prepared in a condition $W_{M}^{0}$ so that the pointer variable $\hat Q_{M}$ is  $\epsilon$-sharp collimated in an interval $I_{0}$ centred at $0$.

In the interaction stage the "particles" $S$ and $M$ interact through a von Neumann impulsive interaction, $\hat H_{I} = \gamma \hat Q_{S}\otimes \hat P_{M}$, see \S \ref{prem},  causing a change in the qr-number value of $\hat Q_{M}$ proportional to the   qr-number value of $\hat Q_{S}$ which doesn't change. When the interact lasts from $t_{1}$ to $t_{2}$, if  the ontic condition of $S$ is $V_{S}^{\pm} \subset W_{S}^{\pm}$ then
 \begin{equation}
q_{M}|_{W_{M}^{0}}(t_{2}) - q_{M}|_{W_{M}^{0}}(t_{1}) = \kappa_{M}q_{S}(V_{S}^{\pm})
\end{equation}where $\kappa_{M} = \frac{\gamma}{m_{M}}(t_{2} - t_{1})$. But $q_{S}(V_{S}^{\pm}) \approx \lambda_{\pm}|_{V_{S}^{\pm}}$ so that the change in the pointer's reading will be proportional to $\lambda_{\pm}.$ In \S \ref{chan}, we show that $q_{M}|_{W_{M}^{0}}(t_{2})  = q_{M}|_{W_{M}^{\pm}} $ where $W_{M}^{\pm}$ is a condition on which $\hat Q_{M}$ is measurable, $W_{M}^{\pm}$ depends on whether  $W_{S}^{\pm}$ was prepared. The problem of definite outcomes does not exist in the qr-number model. The preparation of $S$ ensures that $S$ has an ontic condition that is an open subset of one of the designated epistemic conditions  $W_{S}^{\pm}.$ The outcome for the ensemble is the determination of the relative frequencies, $|\alpha_{+}|^{2}$ and $|\alpha_{-}|^{2}$. This does not determine the wave-function $\psi_{S}(\vec \alpha)$.

The pointer outcome can be amplified, this is discussed in \S \ref{ampl}.

\subsection{Basics of the qr-number  model}\label{basics} The mathematics of the qr-number model, introduced in \cite{14}, is built upon a Hilbert space formalism. It uses a spatial topos, defined in \cite{12} and \cite{13}, to obtain qr-numbers as the numerical values taken by physical attributes of a quantum system.  

In the qr-number model the quantum system always has a complete state, called its condition, given by an open subset of the smooth state space $\EsubS(\mathcal{A}_{S})$, defined in \S \ref{Math},  and all physical attributes retain their qr-number values even when not being observed. The qr-numbers are contextual, the qr-number value of a physical attribute is essentially a function with values in $\RR$ whose domain is the system's condition. 

There are two classes of quantum conditions: (1) the  epistemic condition of an ensemble of systems depends upon the experimental setup and (2) the ontic condition of an individual system in the ensemble. Any open subset of $\EsubS(\mathcal{A}_{S})$ can be in either class but an ontic condition is always proper open subset of an epistemic condition. The existence of ontic conditions explains the variation in the individual outcomes in an experiment. In general a mixed condition of the form $\sum_{j} \lambda_{j} W_{j}$ for $0< \lambda_{j} < 1$, $\sum_{j}\lambda_{j} = 1$ and $W_{j} \in \mathcal{O}(\EsubS(\mathcal{A}_{S}))$ is an epistemic condition, each $\lambda_{j}$ is interpreted as the probability preparing the ensemble in $W_{j}$.  

The physical attributes of a system are represented by the elements of an  O$^{\ast}$-algebra $\mathcal{A}_{S}$, see \cite{9}, of unbounded operators  on a dense subset $\mathcal{D}$ of the system's Hilbert space $\mathcal{H}_{S}$. \footnote{It is not necessary that all attributes are represented in an O$^{\ast}$-algebras, in the Stern-Gerlach experiment spin is represented by bounded operators on $\CC^{2}$.}O$^{\ast}$-algebras allow us to directly represent physical qualities like energy, momentum and position of a particle. When the system is a massive Galilean relativistic quantum particle it has a trajectory in its qr-number space, see \cite{16} and \cite{17} for some examples. In this paper each $O^{\ast}$-algebra comes from a unitary representation $\hat U$ of a Lie group $G$ on $\mathcal{H}$, see \S \ref{Math}. The set of $C^{\infty}$-vectors for $\hat U$, denoted $\mathcal{D}^{\infty}(\hat U)$, is a dense linear subspace of $\mathcal H$ which is invariant under $\hat U(g), \; g \in G$, \cite{19} has more details. The system's smooth state space, $\EsubS(\mathcal{A}_{S}),$ is contained in the convex hull of projections $\scrP$ onto one-dimensional subspaces spanned by unit vectors $\phi  \in \mathcal D$.

\subsubsection{ Qr-number probabilities}\label{qrprob}
The spectral families of self-adjoint operators are used to define quantum probability measures on $\RR$ in \cite{3}. If $ \hat P^{\hat A}(S)$ is the spectral projection operator of $\hat A$ on the Borel subset $S$ of $\RR$, then in the standard interpretation $\mu_{\rho}^{\hat A}(S) = Tr \rho \hat P^{\hat A}(S)$ is the probability that when the system is in the state $\rho$ a measurement of $\hat A$ gives a result in the set $S$.

If the system has the condition $U$, the qr-number probability that $a(U)$ lies in $S$ is $\pi^{\hat A}(S)|_{U}$, the qr-number value of $ \hat P^{\hat A}(S)$ at $U$.

If $U = \nu(\rho_{s}; \delta)$ for $\delta \ll 1$ then, for all Borel sets $S$, $\pi^{\hat A}(S)|_{U} \approx Tr \rho_{s} \hat P^{\hat A}(S) = \mu_{\rho_{s}}^{\hat A}(S)$,  the standard quantum mechanical probability when the system is in the state $\rho_{s}$. so that  $|\pi^{\hat A}(S)|_{U} - \mu_{\rho_{s}}^{\hat A}(S)| < \delta$.

\subsection{ Measurement in the qr-number model}\label{meas}
Measurements are a special class of interactions between two physical systems. The system $\mathcal{S}$ has an attribute, called the measurand, whose value is to be determined. The interaction couples the measurand to a pointer of the measurement apparatus $\mathcal{M}$ whose numerical value can be read. As a result of the interaction the numerical value of $\mathcal{M}$'s pointer is changed by an amount that depends on the value of the measurand which is deducible from the difference of the pointer values. Both $\mathcal{S}$ and $\mathcal{M}$ are assumed to be quantum systems.

\subsubsection{ No measurement is exact}\label{measreq}
The qr-number model accepts that no measurement is exact. In metrology, see \cite{20}, any physical measurement is said to have two components: (1) A numerical value (in a specified system of units) giving the best estimate possible of the quantity measured, and (2) a measure of precision associated with this estimated value. The measure of precision is a parameter that characterises the range of values within which the value of the measurand can lie at a specified level of confidence. The best estimate is quantified by a level of confidence parameter $( 1 - \epsilon)$ in the range $[0,1]$.

The way these parameters are used in the qr-number model is exemplified in \cite{15} by the processes of passing a system $\mathcal{S}$ through a filter. 
The $\epsilon$ sharp collimation of the quantity, represented by $\hat Q_{S}$, in an interval $I \subset \RR$ when the system $\mathcal{S}$ has the condition $W_{S}$ gives a standard real number to approximate the qr-number $q_{S}|_{W_{S}}$.
\begin{definition} For an interval $I$, of width $|I|$, if $W_{S}$ is the largest convex open set in $\EsubS(\mathcal{A}_{S})$ such that $q_{S}|_{W_{S}} \subset I  \; \text{and}  \; (q_{S}^{2}|_{W_{S}} - (q_{S}|_{W_{S}})^{2}) \leq  \frac{\epsilon}{4}|I|^{2}$ then $\hat Q_{S}$ is $\epsilon$ sharp collimated in $I$ on $W_{S}$. \end{definition} Let $\sigma(\hat Q_{S})$ be $\hat Q_{S}$'s  spectrum.  If $W_{S}$ is the condition on which $\hat Q_{S}$ is $\epsilon$ sharp collimated on $I$ and $\exists  \alpha_{0}\in I \cap \sigma(\hat Q_{S})$, then with precision $ |I|/2$ and confidence $(1- \epsilon)$, $\alpha_{0}$ is the measured value of $\hat Q_{S}$. 

\subsubsection{Measurement Conditions}
Conditions that support determining a value of $\hat Q_{S}$ in an interval $I$ are of the form $\mathcal{N}(\hat P_{\phi_{\lambda}},\hat Q_{S}, \delta) = \{\rho \in \EsubS(\mathcal{A}_{S}) : |Tr(\rho \hat Q_{S} - \hat P_{\phi_{\lambda}}\hat Q_{S})|<\delta\}$ where $\phi_{\lambda}$ is an (approximate) eigenstate for some $\lambda \in \sigma_{c}(\hat Q_{S})\cap I,$.. In \cite {15} we prove the following.
\begin{theorem}\label{TH 1}  If $\lambda \in \sigma(\hat Q_{S})$, there exists an interval $I_{\lambda}$ centred on $\lambda$ in which $\hat Q_{S}$ is $\epsilon$ sharp collimation on $\mathcal{N}(\hat P_{\phi_{\lambda}},\hat Q_{S}, \delta) $ for some $\delta > 0.$
\end{theorem} This results builds upon the assumption of standard quantum theory  that the results of measurements are the eigenvalues of the operator which represents the quantity being measured.

A similar result holds for strictly $\epsilon$ sharp collimation, defined in \S \ref{strictesc}, which is used to define qr-number probabilities and to show in \cite{15} that every attribute of $S$ appears to have undergone a L\"uders-von Neumann transformation when the collimation is strictly $\epsilon$ sharp.

In order to complete the determination of a measured value for $\hat Q_{S}$ the system $S$ must interact with a measurement system $M$. The interaction connects the $\epsilon$ sharply collimated qr-number value of $\hat Q_{S}$ with a constant qr-number value of the pointer of $M$ which is observable.

\section{Preparing for a measurement}\label{prep}
There are two ways in which we can describe  the preparation of the system $S$ in the generic example of \S \ref{intro}. In this experiment there is only one attribute $\hat Q_{S}$ to be measured, it has only two eigenvalues $\{\lambda_{s}\}_{ s=\pm}$ whose corresponding eigenvectors $\{\phi_{\lambda_{s}}\}_{s=\pm}$ span a two dimensional subspace $\mathcal{M}_{2}\subset\mathcal{H}_{S}$.

We can use the qr-number model to describe attempts to prepare a state $\hat P_{\psi_{S}(\vec \alpha)}$ where $\psi_{S}(\vec \alpha) = \alpha_{+} \phi_{\lambda_{+}} + \alpha_{-} \phi_{\lambda_{-}} $is a superposition of $\hat Q_{S}$'s eigenstates $\phi_{\lambda_{\pm}}.$ Let $\Gamma = \{\vec \alpha = (\alpha_{+}, \alpha_{-})\in \CC^{2}: |\alpha_{+}|^{2} + |\alpha_{-}|^{2} = 1\}$. The vectors $ \phi_{\lambda_{+}} $ and $ \phi_{\lambda_{-}} $ are orthonormal and if $\vec \alpha \in \Gamma$ then $\psi _{S}(\vec \alpha)$ is normalised. 

The condition $W_{S}^{pure}(\vec \alpha) =  \mathcal{N}(\hat P_{\psi_{S}(\vec \alpha)}, \hat Q_{S}, \delta)$ is centred on the state $\hat P_{\psi_{S}(\vec \alpha)}$. If $W_{S}^{pure} = \cup_{\vec\alpha \in \Gamma} W_{S}^{pure}(\vec \alpha) $ is the prepared epistemic condition of $S$, then $W_{S}^{pure}(\vec \alpha)$ for a particular  pair $\vec\alpha \in \Gamma$ is an ontic condition.

Alternatively assume that the fraction $|\alpha_{+}|^{2}$ of an ensemble is prepared in an epistemic condition $W_{S}^{+} =  \mathcal{N}(\hat P_{\phi_{\lambda_{+}}},\hat Q_{S}, \delta) $, centred on the eigenstate $\hat P_{\phi_{\lambda_{+}}}$, whilst the fraction $|\alpha_{-}|^{2} = 1 - |\alpha_{+}|^{2}$  is prepared in the epistemic  condition $W_{S}^{-} =  \mathcal{N}(\hat P_{\phi_{\lambda_{-}}},\hat Q_{S}, \delta) $, centred on  $\hat P_{\phi_{\lambda_{-}}}.$ The epistemic condition of this ensemble is $W_{S}^{mix}(\vec \alpha) =  |\alpha_{+}|^{2}W_{S}^{+} + |\alpha_{-}|^{2}W_{S}^{-}.$ 

If $\rho_{S}^{mix}(\vec \alpha) = | \alpha_{+}|^{2} \hat P_{\phi_{\lambda_{+}}} +  |\alpha_{-}|^{2} \hat P_{\phi_{\lambda_{-}}}$ then $W_{S}^{mix} = \mathcal{N}(\rho_{S}^{mix}(\vec \alpha) , \hat Q_{S}, \delta)$ as $W_{S}^{\pm} =  \mathcal{N}(\hat P_{\phi_{\lambda_{\pm}}},\hat Q_{S}, \delta).$ 
\begin{lemma} 
\begin{equation}        W_{S}^{pure}(\vec \alpha) =W_{S}^{mix}(\vec \alpha)  = |\alpha_{+}|^{2}W_{S}^{+} + |\alpha_{-}|^{2}W_{S}^{-}.
\end{equation} so the two ways of preparing $S$ produce the same open subset of states.
\end{lemma} $\mathcal{N}(\hat P_{\psi_{S}(\vec \alpha)}, \hat Q_{S}, \delta) = \mathcal{N}(\rho_{S}^{mix}(\vec \alpha) , \hat Q_{S}, \delta)$ because $\phi_{\lambda_{+}} \text{and} \; \phi_{\lambda_{-}} $ are orthogonal eigenvectors of $\hat Q_{S}$, for all $\delta > 0.$ Moreover  $\mathcal{N}(\hat P_{\psi_{S}(\vec \alpha)}, \hat Q_{S}, \delta)$ can be decomposed, see \S \ref{decomp}, as \begin{equation}
 \mathcal{N}(\rho_{S}^{mix}(\vec \alpha) , \hat Q_{S}, \delta) = |\alpha_{+}|^{2} \mathcal{N}(\hat P_{\phi_{\lambda_{+}}},\hat Q_{S}, \delta) +  |\alpha_{-}|^{2} \mathcal{N}(\hat P_{\phi_{\lambda_{-}}},\hat Q_{S}, \delta).\end{equation}
 
The coefficients $\{ |\alpha_{r}|^{2}\}_{r=\pm}$ are the frequency probabilities that when $S$ has the condition $W_{S}(\vec \alpha) = \mathcal{N}(\hat P_{\psi_{S}(\vec \alpha)}, \hat Q_{S}, \delta)$ the attribute $\hat Q_{S}$ is located in intervals $\{I_{\pm}\}$, centred on the eigenvalues $\lambda_{\pm}$. The qr-number probability for location in $I_{+}$ is $ |\alpha_{+}|^{2}\pi^{\hat Q_{S}}(I_{+})|_{W_{S}^{+}},$  \S\ref{qrprob}, because
$\pi^{\hat Q_{S}}(I_{+})|_{W_{S}(\vec\alpha)} = |\alpha_{+}|^{2}\pi^{\hat Q_{S}}(I_{+})|_{W_{S}^{+}}$ as $\pi^{\hat Q_{S}}(I_{+})|_{W_{S}^{-}} =0.$

Therefore, in the fraction $ |\alpha_{+}|^{2}$ of preparation procedures,  $S$ is prepared in an ontic condition $V_{S}(\vec \alpha) \in \mathcal{O}(W_{S}^{+}) $ and, in the fraction $ |\alpha_{-}|^{2}$ of procedures, $S$ is prepared in an ontic condition $V_{S}(\vec \alpha) \in \mathcal{O}(W_{S}^{-}).$ The goal of the experiment is to determine these relative frequencies.

If $\mathcal{M}$ is prepared in a epistemic condition $W_{M}^{0} = \mathcal{N}( \hat P_{\phi_{M}}, \hat Q_{M}, \delta)$ where $\phi_{M}$ is an  eigenstate of the operator $\hat Q_{M}$ for eigenvalue $0$ and $S$ was prepared $W_{S}(\vec \alpha) = \mathcal{N}(\hat P_{\psi_{S}(\vec \alpha)}, \hat Q_{S}, \delta)$. then the combined system has  $W_{S,M}(\vec \alpha, 0) = W_{S}(\vec \alpha)\otimes W_{M}^{0} = \mathcal{N}(\hat P_{\psi_{S}(\vec \alpha)}\otimes \hat P_{\phi_{M}}, \hat Q_{S}\otimes \hat Q_{M}, 2\delta)$ whose central state $\hat P_{\Psi_{S, M}}$ projects onto the product wave-function of the standard model, \begin{equation}\label{EQ5}
\Psi_{S,M}(t_{1}) = \psi_{S}(\vec \alpha)\otimes \phi_{M} = (\alpha_{+} \phi_{\lambda_{+}} + \alpha_{-} \phi_{\lambda_{-}}) \otimes \phi_{M}.\end{equation} However using Lemma 2 the product condition can also be expressed\begin{equation}
W_{S,M}(\vec \alpha, 0)  = \sum_{r =\pm} |\alpha_{r}|^{2} \mathcal{N}(\hat P_{\phi_{\lambda_{r}}}\otimes \hat P_{\phi_{M}}, \hat Q_{S} \otimes \hat Q_{M}, 2\delta)
\end{equation} If $S$ is prepared in an ontic condition $W_{S}(\vec \alpha_{\pm}) \subset W_{S}^{\pm}$ and $M$ in an ontic condition $V_{M} \subset W_{M}^{0}$ then the qr-number values of $\hat Q_{S}$ and $\hat Q_{M}$ are  \begin{equation}\label{EQU9}
q_{S}|_{W_{S}(\vec \alpha_{+})} \approx  \lambda_{+} \;  \text{and}\; q_{S}|_{W_{S}(\vec \alpha_{-})} \approx  \lambda_{-}   \; \text{while}\;q_{M}|_ {V_{M}}\approx 0.
\end{equation} These conditions are such, see theorem \ref{TH1}, that if $\hat Q_{S}$ and $\hat Q_{M}$ were measured at this stage of the experiment, $\hat Q_{S}$  would register a value $ \lambda_{+}$ or $ \lambda_{-}$ and $\hat Q_{M}$ would be $0$.

Now the prepared systems $S$ and $M$ are brought together to interact.

\subsection{ The coupling interaction}\label{prem} The purpose of this interaction is to couple the qr-number value of the measurand $\hat Q_{S}$ to that of the pointer $\hat Q_{M}$ of the measurement apparatus so that a quantitative value can be more easily observed. 

The appropriate interactions include the von Neumann impulsive interactions \cite {7},  Zurek's controlled shifts \cite{22}, as well as Bohm's approximation for the interaction between a magnetic field and the spin of a particle in the Stern-Gerlach experiment \cite {6},$\S 22.6$, and the electric dipole interaction Hamiltonian used in Haroche's Schr\"odinger cat experiment \cite{4}. Each interaction Hamiltonian operator has a similar structure, it features the product of $\mathcal{S}$'s attribute, which is to be measured, with an attribute of $\mathcal{M}$. For example, if $\mathcal{S}$'s attribute is a position operator $\hat Q_{S}$ then $\mathcal{M}$'s attribute will be a momentum operator $\hat P_{M}$ which is conjugate to the position operator $\hat Q_{M}$ for $\mathcal{M}$.\begin{equation}\label{EQU6}
\hat H_{I} = \gamma \hat Q_{S}\otimes \hat P_{M}
\end{equation}  As the interaction is  assumed to be impulsive and the Hamiltonian has only this interaction term, the equations of motion are linear.

The choice of the attributes depends on the physics, for example in the coupling of fields to charges for the Schr\"odinger cat experiment, \cite{4}, a free electron of charge $q$, mass $m$, position $\vec X_{e}$ and momentum $\vec P_{e}$, is coupled to the field which  is described in the Schr\:odinger picture by the vector potential $\vec A(\vec x)$.  If the field is thought of as a quantum system whose spatial locations are labelled by the three components $\{\hat Q^{f}_{j}=A_{j}\}_{j=1}^{3}$ of its vector potential and its momenta by the three components $\{ \hat P^{f}_{j} = E_{j}\}_{j=1}^{3}$ of  its electric field (because $\vec E = -\frac{\partial \vec A}{\partial t}$).  A charge-field interaction  term \begin{equation}\label{EQ8}
\hat H_{af}^{int} = - \frac{q}{m} \hat P_{e} \cdot \hat A(\vec x) = - \frac{q}{m} \hat P_{e} \cdot \hat Q^{f}
\end{equation} is obtained  by neglecting the small magnetic interaction with the electron spin for the Hamiltonian in the Coulomb gauge\footnote{$\vec \nabla \cdot \vec A = 0$ and the scalar potential is negligibly small.} and neglecting a $\vec A^{2}$ term. This interaction is in the form of equation (\ref{EQU6}).

\subsection{The output}\label{output}
Consider the prototypical von Neumann interaction in which the two systems are assumed to be massive one-dimensional quantum particles and the measurand is $\hat Q_{S}$, \cite{7} pp 443. The qr-number value of the interaction Hamiltonian is, during the period $t_{1} < t <t_{2},$\begin{equation}\label{EQ9}
h|_{W} = \gamma q_{S}|_{W}p_{M}|_{W} \end{equation}  where $W = W_{S,M}(\vec \alpha, 0) $ is the prepared product condition. The coupling constant $\gamma$, of dimension $T^{-1}$, is large enough that the kinetic energy can be neglected during the interaction.

The qr-number equations of motion for the position and momentum of $S$ are, see \S \ref{equn}, \begin{equation}\label{EQ10} m_{S}\frac{dq_{S}|_{W}}{ dt} = \frac{\partial h|_{W}}{ \partial p_{S}|_{W}}\; \text{ and} \:\frac{dp_{S}|_{W}}{dt} = - \frac{\partial h|_{W}}{ \partial q_{S}|_{W}}\end{equation} while those for the position and momentum of $M$ are   \begin{equation}\label{EQ11}
m_{M}\frac{dq_{M}|_{W}}{ dt} = \frac{\partial h|_{W}} {\partial p_{M}|_{W}}\; \text{and}\; \frac{dp_{M}|_{W}}{ dt} = - \frac{\partial h|_{W}}{ \partial q_{M}|_{W}}\end{equation}  If the interaction acts over an infinitesimal period $\tau = (t_{2} - t_{1}),$ the qr-number values of $\hat Q_{S}$ and $\hat Q_{M}$ at time $t_{2}$ will be,\begin{equation}\label{EQ12}
 q_{S}|_{W}(t_{2}) = q_{S}|_{W}(t_{1}),\;\; q_{M}|_{W}(t_{2}) = q_{M}|_{W}(t_{1}) + \kappa_{M} q_{S}|_{W}(t_{1}),\end{equation} where for $K = S,M$, $\kappa_{K} = \frac{\gamma\tau}{m_{K}}$ and  $m_{K}$ is the mass. 
 
 If $V = V_{S}\otimes V_{M}$ with $V_{S}\in \mathcal{O}(W_{S}^{+})$ and $V_{M} \in \mathcal{O}(W_{M}^{0})$ by equation (\ref{EQU9}), $q_{S}|_{V}(t_{1}) \approx \lambda_{+}$ so that $q_{M}|_{V}(t_{2}) \approx 0 + \kappa_{M}\lambda_{+}$ and when $V_{S}\in \mathcal{O}(W_{S}^{-})$ and $V_{M} \in \mathcal{O}(W_{M}^{0})$ then  $q_{M}|_{V}(t_{2}) \approx 0 + \kappa_{M}\lambda_{-}$. The standard numbers $\lambda_{+}$ or $\lambda_{-}$ are the measured values because the conditions $ W_{S}^{\pm}$   support $\epsilon$-sharp collimation in intervals centred on $\lambda_{+}$ or $\lambda_{-}$ and $W_{M}^{0}$ supports 
  $\epsilon$-sharp collimation in an interval centred on $0$.
  
 Thus the difference between the measurement pointer readings is proportional to $|\lambda_{+} - \lambda_{-}|$ which is observable  when the eigenvalues are sufficiently separated. This resolves the problem of definite outcomes in the qr-number approach.

\subsection{How the conditions changed}\label{chan} It is interesting to see how the conditions changed during a measurement, the analysis is closer to that of standard quantum theory.
 
Before interacting, at time $t_{1},$ the joint condition was  $ W_{S,M}(t_{1}) = \mathcal{N}(\hat P_{\Psi_{S,M}}(t_{1}), \hat Q_{S}\otimes \hat Q_{M}, 2\delta)$  where $\Psi_{S,M}(t_{1}) = \psi_{S}(\vec \alpha) \otimes \phi_{M} = (\alpha_{+} \phi_{\lambda_{+}} + \alpha_{-} \phi_{\lambda_{-}}) \otimes \phi_{M}$ and $ |\alpha_{+}|^{2} + |\alpha_{-}|^{2} =1.$ The initial qr-numbers values of $S$'s attribute $\hat Q_{S}$ is $(q_{S}|_{W_{S,M}})(t_{1})= q_{S}|_{W(S)(\vec \alpha)}$,  while $M$'s attribute has   $(q_{M}|_{W_{S,M}})(t_{1})=q_{M}|_{W(M)^{0}}$ as its initial qr-number value.

 Using the evolution of the conditions, discussed in \S \ref{evolcon},  at time $t_{2},$ after the coupling interaction, the condition has evolved to  $W_{S,M}(t_{2}) = \mathcal{N}(\hat P_{\Psi_{S,M}}(t_{2}), \hat Q_{S}\otimes \hat Q_{M}, 2\delta)$ where $\Psi_{S,M}(t_{2}) = (\alpha_{+} \phi_{\lambda_{+}}\otimes \phi_{M}^{+} + \alpha_{-} \phi_{\lambda_{-}} \otimes \phi_{M}^{-})$ with orthonormal vectors  $ \phi_{M}^{+}, \phi_{M}^{-} \in \mathcal{H}_{M}.$ The $\{ \phi_{M}^{s}\}_{s=\pm}$ are (approximate) eigenvectors for $\hat Q_{M}$'s (continuous) spectrum, assumed to be orthogonal.  Then \begin{equation}
\mathcal{N}(\hat P_{\psi_{S, M}(t_{2})}, \hat Q_{S}\otimes \hat Q_{M}, 2\delta) = \mathcal{N}(\rho_{S,M}^{mix}, \hat Q_{S}\otimes \hat Q_{M}, 2\delta).
\end{equation} where \begin{equation}\rho_{S,M}^{mix}= |\alpha_{+}|^{2}\hat P_{\phi_{\lambda_{+}}}\otimes \hat P_{\phi_{M}^{+}} +  |\alpha_{-}|^{2}\hat P_{\phi_{\lambda_{-}}}\otimes \hat P_{\phi_{M}^{-}}.\end{equation} If we define $W_{S,M}^{+}(t_{2}) = \mathcal{N}(\hat P_{\phi_{\lambda_{+}}}\otimes \hat P_{\phi_{M}^{+}}, \hat Q_{S}\otimes \hat Q_{M}, 2\delta) $ and $W_{S,M}^{-}(t_{2}) = \mathcal{N}(\hat P_{\phi_{\lambda_{-}}}\otimes \hat P_{\phi_{M}^{-}}, \hat Q_{S}\otimes \hat Q_{M}, 2\delta) $ then  \begin{equation}
 \mathcal{N}(\rho_{S,M}^{mix}, \hat Q_{S}\otimes \hat Q_{M}, 2\delta) = |\alpha_{+}|^{2}W_{S,M}^{+}(t_{2}) +  |\alpha_{-}|^{2}W_{S,M}^{-}(t_{2}) . 
\end{equation}

After the interaction the qr-numbers values of $\hat Q_{S}$ and $\hat Q_{M}$ are 
\begin{equation}
q_{S}|_{W_{S,M}(t_{2})}=  |\alpha_{+}|^{2}q_{S}|_{W_{S}^{+}}+  |\alpha_{-}|^{2}q_{S}|_{W_{S}^{-}}\end{equation} and \begin{equation}q_{M}|_{W_{S,M}(t_{2})} =  |\alpha_{+}|^{2}q_{M}|_{W_{M}^{+}}+  |\alpha_{-}|^{2}q_{M}|_{W_{M}^{-}}.\end{equation} 
 
The solution of the qr-number equations of motion are given in equation (\ref{EQ12}) of \S \S\ref{output}, the first expression equates the qr-number value of $\hat Q_{S}$ at the end of the interaction to its value at its commencement. Initially  $S$'s reduced condition  $W_{S}(t_{1}) = \mathcal{N}( \hat P_{\psi_{S}(\vec \alpha)}, \hat Q_{S},\delta)$ is centred on the state $\hat P_{\psi_{S}(\vec \alpha)}$,  at the end $S$'s reduced condition $W_{S}(t_{2}) = \mathcal{N}(\rho_{S}^{mix}, \hat Q_{S}, \delta) $ is centred on the mixed state $\rho_{S}^{mix} =  |\alpha_{+}|^{2} \hat P_{\psi_{S}^{+}} + |\alpha_{-}|^{2} \hat P_{\psi_{S}^{-}} $. As was shown in \S \ref{prep},  \begin{equation}
\mathcal{N}(\hat P_{\psi_{S}(\vec\alpha)}, \hat Q_{S}, \delta) = \mathcal{N}(\rho_{S}^{mix}9\vec\alpha), \hat Q_{S}, \delta).
\end{equation} so that $W_{S}(t_{2}) = W_{S}(t_{1}).$

 The change in $M$'s pointer reading expressed in the second expression of equation (\ref{EQ12}) has been discussed in \S\ref{output}. If $W_{S}(t_{2})$ and $W_{M}(t_{2})$ are conditions reduced from $W_{S,M}(t_{2})$ then \begin{equation}\label{EQ14}
 q_{M}|_{W_{M}}(t_{2}) -  q_{M}|_{W_{M}}(t_{1}) = \kappa_{M}  q_{S}|_{W_{S}}(t_{1})
\end{equation} shows how definite outcomes are obtained, as $\hat Q_{S}$ is respectively $\epsilon$-sharp collimated in  $I_{S}^{\pm}$ centred on $\lambda_{\pm}$ on the conditions  $W_{S}(t_{1}) \subset W_{S}^{\pm}$ while $\hat Q_{M}$ is $\epsilon$-sharp collimated in  $I_{M}^{\pm}$ centred on $\kappa_{M} \lambda_{\pm}$ on  $W_{M}(t_{2}) \subset W_{M}^{\pm}$ and  is $\epsilon$-sharp collimated in  $I_{M}^{0}$ centred on $0$ on  $W_{M}(t_{1}) \subset W_{M}^{0}.$

 If we wish to measure the momentum of a system the prototype would use a von Neumann implusive interaction whose labels were interchanged as in equation(\ref{EQ9}) then \begin{equation}
h_{I}|_{W} = \gamma\;  p_{S}|_{W}q_{M}|_{W},\end{equation} with $\gamma$ the coupling constant, $\hat Q_{M}$ is $\mathcal{M}$'s position operator and $\hat P_{S}$ is $\mathcal{S}$'s momentum operator whose value is to be measured. A similar set of outcomes when $\hat P_{S}$ is $\epsilon$-sharp collimated follows the obvious changes.

\subsection{Amplification of the output}\label{ampl} Consider a chain of couplings between a sequence of outputs and measurement systems each of which augments the magnitude of the next output. The component of the measurement apparatus that initially interacts with the system $S$ will be denoted $M_{0}$.The output $ q_{M_{0}}|_{W_{M_{0}}}(t_{2}) $ is the input for a second von Neumann interaction between the attributes $\hat Q_{M_{0}}$ and $\hat P_{M_{0}}$ of the first component and attributes  $\hat Q_{M_{1}}$ and $\hat P_{M_{1}}$ of the next.

For the $k^{th}$ link in this chain of events, the input is denoted $ q_{M_{k-1}}|_{W_{M_{k-1}}} $ and the output is $ q_{M_{k}}|_{W_{M_{k}}} $. Here $W_{M_{k-1}} \subset W_{M_{k-1}}^{\pm}$ so that the interaction at  the $k^{th}$ stage is \begin{equation}\label{EQ19}
h|_{W_{M_{k-1}} } = \gamma q_{M_{k-1} }|_{W_{M_{k-1}} }p_{M_{k}}|_{W_{M_{k-1}} } \end{equation}  The interaction is  assumed to be impulsive and only acting between $t_{k-1}$ and $t_{k},$ then  at $t = t_{k},$ the qr-number value of $\hat Q_{M_{k}}$ is 
\begin{equation}\label{EQ24}
 q_{M_{k}}|_{W_{M_{k}}}(t_{k}) -  q_{M_{k}}|_{W_{M_{k}}}(t_{k-1}) = \kappa_{M_{k}}  q_{M_{k-1}}|_{W_{M_{k-1}}}(t_{k-1})
\end{equation}

When the pointer is linked via impulsive interactions to the parts $\{M_{l}\}_{l=0}^{k}$ and $M_{0} = M$,  then the location after the $k^{th}$ interaction is changed by \begin{equation}
(\prod_{l=0}^{k} \kappa_{M_{l}}) |q_{S}|_{W_{S}(t_{1})}|
\end{equation} Thus the output is amplified if each $\kappa_{M_{l}} > 1$.

\section{Appendices}
\subsection{ Mathematics of qr-numbers}\label{Math}
The qr-number value of a physical quantity depends not only on the operator $\hat A$ that represents it but also on the condition of the system. They differ from standard real numbers that are represented in the qr-number model by globally constant qr-numbers. For a summary of the mathematical structure of the qr-number model, see Corbett\cite{19}. 

When a system $\mathcal{S}$ has a Hilbert space $\mathcal{H}_{S}$ that carries a unitary representation $U$ of a symmetry group $G$ then its physical attributes are represented by operators that form an O$^{\ast}$-algebra $\mathcal{A}_{S}$: the representation $dU$ of the enveloping algebra $\mathcal{E}(\mathcal{G})$ of the Lie algebra $\mathcal{G}$ of $G$  see \cite{9}. The operators have a common domain $\mathcal{D} = \mathcal{D}^{\infty}(U)$, the set of $C^{\infty}$-vectors for the representation $U$. 
 \begin{definition}\label{D1}
The states on $\mathcal A_{S} $ are the strongly positive 
linear functionals on $\mathcal A_{S}$ that are normalised to take the 
value $1$ on the unit element $\hat I $ of $\mathcal A_{S} $, they form the state space $\EsubS (\mathcal A_{S})$.
\end{definition}$\EsubS(\mathcal A)$ has the weak topology generated by the functions $ a(\cdot)$ where, given $\hat A \in  \mathcal A_{S},  a( \rho) = Tr \hat A  \rho,  \forall   \rho  \in \EsubS(\mathcal {A}_{S}) $. This topology is the weakest that makes all the functions $ a(\cdot)$ continuous.
For $\hat A \in \mathcal A_{S}, \epsilon > 0$ and $\rho_0 \in \EsubS(\mathcal{A}_{S})$, the sets  $ \scrN( \rho_{0}  ; \hat A ; \epsilon) = \{ \rho \ ; |Tr \rho \hat A - Tr \rho_{0} \hat A| < \epsilon \}$  form an open 
sub-base for the weak topology on $\EsubS(\mathcal A)$. The basic open subsets are denoted  
$\nu( \rho_{1} ; \delta) = \{  \rho \; :Tr | \rho -  \rho_{1}| < \delta\}$.  $\EsubS(\mathcal{A}_{S})$ is compact in the weak topology\cite{18}. 
\begin{definition}\label{D2}
A trace functional on $\mathcal A$ is a functional of the form $ \hat A \in \mathcal A \mapsto Tr (\hat B \hat A) $ for some trace class operator $\hat B$.
\end{definition}
\begin{theorem}\cite{9}\label{TH1}
Every strongly positive linear
functional  on $\mathcal A $ is given by a trace functional.
\end{theorem}

\subsubsection{Locally linear qr-numbers} are denoted $\AAA(\EsubS(\mathcal A))$.\begin{definition}
Let $U\in \mathcal{O}(\EsubS(\mathcal A))$,  a  function $f : U \to \RR$ is {\em locally linear} if each $\rho \in U$ has an open neighborhood $U_{\rho} \subset U$ with an essentially self-adjoint operator $\hat A \in \mathcal A$ such that $f|_{U_{\sigma} } = a(U_{\sigma} )$ for every $\sigma \in U$.
\end{definition}

Density: Given any qr-number $f$ on $U\in \mathcal{O}(\EsubS(\mathcal A))$ and any integer $n$ there exists an open cover $\{U_{j}\}$ of $U$ with for each $j$ a locally linear function $g_{j} : U_{j} \to \RR$ such that $|f|_{U_{j}} - g_{j}(U_{j})| < \kappa/n$, where $\kappa< \infty $ has the same physical dimensions as $f$ and $g$. This means that every qr-number is a union of locally linear qr-numbers, $f(U) = \cup_{j}g_{j}(U_{j})$. 

\subsubsection{Infinitesimal qr-numbers}
The relationship of the qr-number equations of motion with the standard quantum mechanical equations is obtained using infinitesimal qr-numbers. In the following $\mathcal A $ is assumed to be the representation of the enveloping Lie algebra $d\hat U(\mathcal{E}(\mathcal{G}))$ obtained from the unitary representation of the Lie group $G$. 

Infinitesimal qr-numbers are the difference between neighbouring qr-numbers.  Two qr-numbers $x$ and $y$ are neighbours if they are not identical but they do not satisfy $ x > y \lor x < y $ on any non-empty open subset of $\EsubS(\mathcal A )$. The difference $(x-y)$ between neighbouring numbers is an order theoretical infinitesimal number because there is no open set on which $(x - y)>0 \lor (x -y)<0 $ is true. Since qr-real numbers do not satisfy trichotomy the difference between neighbouring real numbers is not zero.

For example: if $V_{0} = \nu (\rho_{0};\delta )$, for $\rho_{0}\in \EsubS(\mathcal A )$ and $\delta >0$, consider a depleted open set $\tilde V_{0} = V_{0}\setminus \{\rho_{0}\}$. Then $q_{j}(V_{0}) $ and $q_{j}(\tilde V_{0})$ are neighbouring qr-numbers because  \begin{equation}
q_{j}(V_{0}) \neq q_{j}(\tilde V_{0})\; \text{and neither} \; q_{j}(V_{0}) > q_{j}(\tilde V_{0})\;\text{nor}\;q_{j}(\tilde V_{0}) > q_{j}(V_{0})
\end{equation} on any open subset of $ \EsubS(\mathcal A )$. 

In fact, $q_{j}(V_{0}) - q_{j}(\tilde V_{0}) = q_{j}^{0}(\rho_{0}) = Tr \rho_{0}\hat Q_{j}$. Since the singleton set $\{\rho_{0}\}$ has empty interior, there is no non-empty open set $W$ on which the difference is non-zero. Thus the expectation values of  quantum mechanical operators are order theoretic infinitesimal qr-numbers.  They are also algebraic infinitesimal qr-numbers because there is no non-empty open set on which the square is non-zero, for $(q_{j}(V_{0}) - q_{j}(\tilde V_{0}))^{2} = q_{j}(V_{0})^{2} - q _{j}(\tilde V_{0})^{2} =  (Tr \rho_{0}\hat Q_{j})^{2}$, which is only non-zero at $\rho_{0}$.

The expectation values $Tr \rho \hat A$ are infinitesimal linear qr-numbers for any state $\rho \in  \EsubS(\mathcal A)$ and any self adjoint operator $\hat A$ in the algebra $\mathcal A$. They are part of the infinitesimal structure of the qr-number world. 

\subsection{Preparation processes.}\label{measint}During a preparation
process a number of quantities are treated successively. One of $\mathcal{S}$'s attributes, represented by the self-adjoint operator $\hat A$, is strictly $\epsilon$-sharp collimated in the interval $ I= \; ]a_1, a_2 [$ when $\mathcal{S}$ has the condition $U$ and immediately afterwards a second attribute, represented by a self-adjoint operator $\hat B$, compatible with $\hat A$ (that is they strongly commute), is strictly $\epsilon$-sharp collimated  in the interval $J = ]b_1 , b_2 [$ when $\mathcal{S}$ has the condition $W$. The qr-number values of $\hat A$ and $\hat B$ will persist with a probability greater than $(1 -\epsilon)$.\cite{15} \S III A. The temporal
order in which the qr-number values were prepared does not affect their values. The system ends up in a a condition $U\cap W$. This extends to finite sets $\{\hat A_{j}\}_{j=1}^{n}$ of commuting operators in the obvious way. If the attributes, represented by the operators $\{\hat A_{j}\}_{j=1}^{n},$ are each $\epsilon$ sharp collimated in their respective intervals $\{I_{j}\}_{j=1}^{n}$ on conditions $\{W_{j}\}_{j=1}^{n}$ then if $\{\alpha_{j}\}_{j=1}^{n}$ are the midpoints of the intervals,  we can, with precision $ |I_{j}|/2$ and confidence $(1- \epsilon)$, take $\alpha_{j}$ to be the classical value of the quantity represented by $\hat A_{j}$ when the system has the condition $\cap_{j=1}^{n} W_{j}$. This an epistemic condition, any open subset of $\cap_{j=1}^{n} W_{j}$ may be the ontic condition of an individual system in the ensemble.

It can be extended to attributes represented by operators that don't commute. Heisenberg's uncertainty relations limit the precision of the simultaneous measurements of the attributes but do not prohibit their measurement, \cite{15} \S C, Theorem 2. For example, a particle's position $\hat Q$ and momentum $\hat P$ satisfy $\imath [ \hat P, \; \hat Q] = \hbar$, so that if the particle with the condition $W$ has both $\hat Q$ and $\hat P$ $\epsilon$-sharp collimated in intervals $I_{q}$ and $I_{p}$ with precisions $ \kappa_{q}$ and  $\kappa_{p}$ then 
$ \kappa_{q} \kappa_{p} \geq \frac{\hbar} {2\epsilon}$ and the product of the intervals' widths satisfy $ |I_{q}| |I_{p}| \geq \frac{2\hbar}{\epsilon}.$ The precisions of the measured values are thus restricted by the inequality $ \kappa_{q} \kappa_{p} \geq \frac{\hbar} {2\epsilon}$.

\subsubsection{More on $\epsilon$ sharp collimation}\label{strictesc}
Recall the definition of $\epsilon$ sharp collimation, \begin{definition} For an interval $I$, of width $|I|$, if $W_{S}$ is the largest convex open set in $\EsubS(\mathcal{A}_{S})$ such that $q_{S}|_{W_{S}} \subset I  \; \text{and}  \; (q_{S}^{2}|_{W_{S}} - (q_{S}|_{W_{S}})^{2}) \leq  \frac{\epsilon}{4}|I|^{2}$ then $\hat Q_{S}$ is $\epsilon$ sharp collimated in $I$ on $W_{S}$. \end{definition} 

On the other hand the qr-number value of an attribute, $\hat A$, can be weakly or strongly contained in an interval. Let $\mathcal{S}$ have the condition $W$, then $\hat A$ lies weakly in an interval $I_{a} \subset \RR$  if  the range of $a|_{W} \subseteq I_{a}$. Using the qr-number value $\pi^{\hat A}(I_{a})|_{W}$ of $\hat A$'s spectral projection operator $\hat P^{\hat A}(I_{a})$ for $I_{a}$, we say that $a|_{W}$ lies strongly in $I_{a}$ when it lies weakly in $I_{a}$ and  $(1 - \epsilon) < \pi^{\hat A}(I_{a})|_{W}  \leq  1$\footnote{The qr-number $\pi^{\hat A}(I_{a})|_{W}$ can be interpreted\cite{15} as the qr-number probablity of the system passing through the slit $I_{a}$, then  $\epsilon$ sharp location in the interval $I_{a}$ requires the qr-number probability to be greater than $(1-\epsilon).$ }. $\hat A$ is then said to be $\epsilon$ sharp located in the interval $I_{a}$ on the condition $W$\cite{15}.
 
 The following result was proven in \cite{15}, 
\begin{theorem}\label{TH2}If  $\hat A$ is $\epsilon$ sharp  collimated in $I_{a}$ on $W$, then $\hat A$ is $\epsilon$ sharp located in $I_{a}$ on $W$.\end{theorem}

Strictly $\epsilon$ sharp collimation is a stronger version of $\epsilon$ sharp collimation that also uses the spectral projection operator,$ \hat P^{\hat A}(I)$, for $\hat A$ on $I$. It requires that $W$ is such that  the qr-number $a|_{W}$ closely approximates the qr-number value $pap|_{W}$ of $\hat P^{\hat A}(I)\hat A\hat P^{\hat A}(I)$.
\begin{definition} $\hat A$ is strictly $\epsilon$ sharp collimated in $I$ on $W$ if it is $\epsilon$ sharp collimated on $W$ and for all $\rho \in W$, $ Tr | \rho - \hat P^{\hat A}(I) \rho \hat P^{\hat A}(I) | < \epsilon$.\end{definition} 
When the O$^{\ast}$-algebra $\mathcal{A}$ is the infinitesimal representation $d\hat U$ of the enveloping algebra $\mathcal{E}(\mathcal{G})$ obtained from a unitary representation $\hat U$ of a Lie group $G$ this suffices because for all $\rho \in W$   \begin{equation}\label{EQ221}
|Tr \rho ( \hat A - \hat P^{\hat A}(I)\hat A \hat P^{\hat A}(I)) |\leq \kappa_{m}(\hat A) Tr | \rho - \hat P^{\hat A}(I) \rho \hat P^{\hat A}(I) | 
\end{equation} where $\kappa_{m}(\hat A) = \sup_{\psi \in \mathcal{D}^{\infty}(\hat U)} \| \hat A d\hat U((1- \Delta)^{m})^{-1} \psi \|/ \|\psi\| < \infty$  with $\Delta = \sum_{i=1}^{d} x_{i}^{2}$ is the Nelson Laplacian in $\mathcal{E}(\mathcal{G})$  with basis $\{ x_{1},x_{2},......,x_{d}\}$ and integer $m>0$.  Thus if $Tr | \rho - \hat P^{\hat A}(I) \rho \hat P^{\hat A}(I) | < \epsilon$ for all $\rho \in W$ then $|a|_{W} - (pap)|_{W} | < \kappa_{m}(\hat A) \epsilon$. In \cite{22}, \S 5.5, for $j =1,2,3$, it is shown that if $G$ is the Weyl-Heisenberg group,  $\kappa_{1}(\hat Q_{j}) = \kappa_{1}(\hat P_{j}) = \frac{1}{2}$.

The next theorems reveal that when $\alpha_{0}$ is in the spectrum of $\hat A$ the condition for strictly $\epsilon$ sharp collimation is a basic open set centred on the eigenstate for $\hat A$ at $\alpha_{0}$, they are proven in \cite{15}.
\begin{theorem} If $\alpha_{0} \in\sigma(\hat A)\cap I$ and $\rho_{0} = |\psi_{0}\rangle\langle \psi_{0}|$ is an eigenstate of $\hat A$ at $\alpha_{0}$ with $\hat P^{\hat A}(I_{a}) \rho_{0}\hat P^{\hat A}(I_{a})  = \rho_{0}$, then  $\forall \epsilon > 0, \exists \delta > 0$  such that $\hat A$ is strictly $\epsilon$ sharp collimated in $I_{a}$ on $\nu(\rho_{0}, \delta)$ and on $\mathcal{N}(\rho_{0}, \hat Q_{S}, \frac{\delta}{2}).$ \end{theorem}

There is an analogous result for the interval $I_{a}$ with midpoint  $\alpha_{0}$ is in the continuous spectrum of $\hat A$.
\begin{theorem} If $\alpha_{0} \in\sigma_{c}(\hat A)$, the continuous spectrum of $\hat A$, and $\rho_{0} = |\psi_{0}\rangle\langle \psi_{0}|$ is an approximate eigenstate of $\hat A$ at $\alpha_{0}$ at accuracy $\delta_{0}$ and $\hat P^{\hat A}(I_{a}) \rho_{0}\hat P^{\hat A}(I_{a})  = \rho_{0}$, then  $\forall \epsilon > 0, \exists \delta > 0$  such that $\hat A$ is strictly $\epsilon$ sharp collimated in $I_{a}$ on $\nu(\rho_{0}, \delta)$ and on $\mathcal{N}(\rho_{0}, \hat Q_{S}, \frac{\delta}{2}).$\end{theorem}

\section{ qr-number equations of motion for massive particles.}\label{equn}
The motion of microscopic particles is governed by equations which have the same form as those for macroscopic particles with qr-numbers  replacing standard real numbers, \cite {14}.   

The laws of motion for a particle of mass $m>0$ are Hamiltonian equations of motion expressed in qr-numbers; $m\frac{dq_{j}|_{U}}{ dt} = \frac{\partial h|_{U}} {\partial p_{j}|_{U}}$ and
${dp_{j}|_{U}\over dt} = - \frac{\partial h|_{U}}{\partial q_{j}|_{U}}$, where $q_{j}|_{U}$, $p_{j}|_{U}$ and $h|_{U}$ are qr-number values of the $j$th components of its position, momentum and of the Hamitonian at the condition $U$.
Thus, if $ h(\vec q|_{U}(t),\vec p|_{U}(t))  = \sum_{j=1}^{3} \frac{1}{2m}(p_{j}|_{U}(t))^{2} + V(\vec q|_{U}(t)) $ is the qr-number value of the Hamiltonian
\begin{equation}\label {EQ25}
\frac{dq_{j}|_{U}(t)}{dt} =  \frac{\partial h(\vec q|_{U}(t),\vec p|_{U}(t))}{\partial p_{j}|_{U}(t)} = \frac{1}{m}p_{j}|_{U},
\end{equation}
\begin{equation}\label{EQ26} \frac{dp_{j}|_{U}(t)}{dt} = - \frac{\partial h(\vec q|_{U}(t),\vec p|_{U}(t))}{\partial q_{j}|_{U}(t)} = f_{j}(\vec q|_{U}).
\end{equation}
 The force has components $f_{j}(\vec q|_{U}) = - \frac{\partial V(\vec q|_{U}(t))}{\partial q_{j}|_{U}(t)}.$

When $\hat A \in \mathcal {A}$ and the time derivative of its qr-number $a(\vec q|_{U}, \vec p|_{U})$ is taken along a trajectory of the particle, then \begin{equation}\label {EQ27}
\frac{da}{dt} = [a,h] \equiv \sum ( \frac{\partial a}{\partial q_{i}} \frac{\partial h}{\partial p_{i}} -  \frac{\partial a}{\partial p_{i}} \frac{\partial h}{\partial q_{i}}).
\end{equation} If the time $t$ occurs explicitly in $a$, $\frac {\partial a}{\partial t}$ must be added to $ [a,h]$.
The bracket $[a,h]$ is the Poisson bracket of the functions $a(\vec q|_{U}, \vec p|_{U})$ and $h(\vec q|_{U}, \vec p|_{U})$. 
The qr-number equation $\frac{da}{dt} = [a,h] $ is the basic dynamical equation for the evolution of the qr-number values of attributes.

\subsection{ Infinitesimal qr-number equations of motion}\label{infeq} In \cite{19}, using approximate eigenvectors for numbers in the continuous spectra of the commuting operators $\{\hat Q_{j}\}_{j=1}^{3}$ when the force operators $\hat F_{j} = f_{j}(\hat Q_{1},\hat Q_{2},\hat Q_{3})$, for $j =1,2,3$, belong to the algebra $\mathcal{A}$, the standard quantum mechanical equations of motion for a massive particle are obtained from linear infinitesimal qr-number approximations to the qr-number Hamiltonian equations of motion, equations (\ref{EQ11}) and (\ref{EQ12}).

When the operators $\{\hat Q_{j}\}_{j=1}^{3}$ have only continuous spectra, for all $\rho \in   \EsubS(\mathcal A)$ and any $\epsilon >0$, \begin{equation}|Tr \rho \hat F_{l} - f_{l}(Tr \rho\hat Q_{1},Tr\rho \hat Q_{2},Tr \rho \hat Q_{3})| < \epsilon, \;\text{for}\; l = 1,2,3.
\end{equation} Therefore for all states $ \rho \in \mathcal{E}_{S}(\mathcal{A})$,  the linear qr-number approximations to the qr-number equations of motion yield the infinitesimal qr-number equations, \begin{equation}\label{EQ29}
 \frac{d}{dt}Tr \rho\hat Q_{j} = \frac{1}{m}Tr \rho \hat P_{j}\; \text{and}\; \frac{d}{dt}Tr \rho \hat P_{j} = Tr \rho \hat F_{j}, \;\text{for}\; j = 1,2,3,
\end{equation} from which Heisenberg's operator equations follow on the assumption that all the time dependence is carried by the operators. If all the time dependence were carried by the states and we assume that $Tr \rho_{t} \hat A = Tr \rho \hat A_{t}$ holds for all operators $\hat A \in \mathcal{A}_{S}$ then it is possible that the time dependence of the states is unitary,
$\hat A_{t} = \hat U_{t} \hat A \hat U_{t}^{-1}$. A unitary evolution of the conditions is compatible with the infinitesimal qr-number equations.

 In the following the conditions can be ontic or epistemic.
 
\subsection{The evolution of the conditions}\label{evolcon} The unitary evolution of the states is compatible with the infinitesimal qr-number equations, see \S \ref{infeq}, so that a condition evolves following the unitary evolution of its component states, that is, if $\rho \to \rho_{t} = \hat U_{t} \rho \hat U_{t}^{\ast}$ for all $\rho \in W$ then $ W \to W_{t} =   \hat U_{t} W \hat U_{t}^{\ast}$. Since  the open sets $\{\nu(\rho, \delta)\}$ are basic in the topology on $ \EsubS(\mathcal{A})$, it suffices to show that for any $\delta >0$, $\nu(\rho, \delta) \to  \nu(\rho_{t}, \delta)$.

\begin{lemma}\label{L2} If $\rho, \rho^{\prime}\in \EsubS(\mathcal{A})$ then $Tr|\rho_{t} - \rho^{\prime}_{t}| = Tr|\rho - \rho^{\prime}|$ when $\forall \rho \in \EsubS(\mathcal{A})\;,\rho_{t} = \hat U_{t} \rho \hat U_{t}^{\ast}$ for a unitary group $\{ \hat U_{t}; t \in \RR \}$, thus if $\rho_{0}\to \rho_{t}$ then $\nu(\rho_{0}, \delta) \to  \nu(\rho_{t}, \delta)$ for any $\rho_{0}$ and any $\delta >0$.
\end{lemma} The proof uses $| \hat U_{t} (\rho - \rho^{\prime} ) \hat U_{t}^{\ast}| = \hat U_{t}|\rho - \rho^{\prime} | \hat U_{t}^{\ast}$ and that the trace is independent of the orthonormal basis used in its evaluation.

\section{Conditions for two systems}
The combined conditions are product conditions when  $S$ and $M$ are not interacting. Each system has its own attributes, represented by O$^{\ast}$-algebras $\mathcal{A}_{S}$ and $\mathcal{A}_{M}$, defined on dense subsets $\mathcal{D}_{S}$ and $\mathcal{D}_{M}$ of their Hilbert spaces $\mathcal{H}_{S}$  and $\mathcal{H}_{M}$ with smooth state spaces $\EsubS(\mathcal{A}_{S})$ and  $\EsubS(\mathcal{A}_{M})$. The attributes have independent qr-number values. 
\begin{definition}\label{D3}
A condition $W = W_{S,M}$ is a product condition  with respect to the decomposition into systems $S$ and $M$ if for every pair of physical attributes, $\hat A_{S}\otimes \hat I_{M}$ of $S$ and $\hat I_{S}\otimes \hat B_{M}$ of $M$, the qr-number  value of $\hat A_{S}\otimes \hat B_{M}$  is a product \begin{equation}
(a_{S}\otimes b_{M})|_{W_{S,M}} = a_{S}|_{\tilde W_{S}}b_{M}|_{\tilde W_{M}}  
\end{equation} where $\tilde W_{S}$ and $\tilde W_{M}$ are the reduced conditions for $S$ and $M$ respectively.\end{definition}Before they interact every state of the combined system is a product state so that every condition is a product condition. 

If $S$  was prepared in a mixed condition $W_{S}^{m} = \sum_{j =1}^{N} p_{j} W_{S}^{j}$, with the condition $W_{S}^{j}$ occurring with probability $p_{j}$, while the condition $W_{M}$  was held fixed for $M$, the ensuing combined condition is still a product condition as $(a_{S}\otimes b_{M})|_{W_{S,M}} =  a_{S}|_{W_{S}^{m}} b_{M}|_{W_{M}}$  for all $\hat A_{S}$ and $\hat B_{M}.$

On the other hand there are entangled conditions, produced when the systems are interacting. \begin{definition}\label{D4}
$W = W_{S,M}$ is an entangled condition if it is not a product condition, i.e., if there is at least one pair of attributes, $\hat A_{S}\otimes \hat I_{M}$ of $S$ and  $\hat I_{S}\otimes \hat B_{M}$ of $M$  such that the qr-number value of $\hat A_{S}\otimes \hat B_{M}$ is not a product
\begin{equation}(a_{S}\otimes b_{M})_{W} \neq a_{S}|_{\tilde W_{S}}b_{S}|_{\tilde W_{2}}.  
\end{equation} \end{definition} A product condition for the combined system before the interaction can evolve into an entangled condition during the interaction, in the same way as product states evolve into entangled states.

Since relations that hold between qr-numbers at a condition $W$ hold on all open subsets $V \subset W$, 
if an epistemic condition $W$ is entangled it has no open subset $V\subset W$ that is a product condition and if $W$ is a product condition then so also is every open subset $V \subset W$.

Finally, a separable mixed condition is prepared  if, while preparing a mixed condition  for $S$, whenever a  $W_{S}^{j}$ is prepared for $S$ a companion condition $W_{M}^{j}$ is prepared for $M$. Then $W_{S,M}^{sep} = \sum_{j =1}^{N} p_{j} W_{S}^{j} W_{M}^{j}$
 so that  \begin{equation}(a_{S}\otimes b_{M})|_{W_{S,M}^{sep}} =  \sum_{j =1}^{N} p_{j} a_{S}|_{W_{S}^{j}} b_{M}|_{W_{M}^{j}}.\end{equation} Such a combined condition is not a product condition nor is it an entangled condition, the outcomes are correlated which is explainable in terms of its preparation at the classical probabilities $p_{j}$.
 
\subsection{Reduced conditions}\label{redcon} For non-identical massive Galilean invariant particles, let $( \mathcal{H}(1,2),\mathcal{A}(1,2),\EsubS(\mathcal{A}(1,2))$ represent a two particle system's Hilbert space, its algebra of physical attributes, and smooth state space  with $ \mathcal{H}(1,2) = \mathcal{H}(1)\otimes \mathcal{H}(2)$ and $\mathcal{A}_{1,2} = \mathcal{A}_{1}\otimes \mathcal{A}_{2}$. 

If $W(1,2)$ is a two particle condition then, for $j = 1,2$, the reduced single particle conditions $\tilde W(j),\; j=1,2$ are obtained by tracing over an  orthonormal basis of the Hilbert space $\mathcal{H}(k)$ for $k = 1\lor2 \neq j$, a straight forward calculation in \cite{16} yielded 
\begin{proposition}\label{TH243} If $ \rho_{0}(1,2) = \hat P_{\phi_{R}(1)} \otimes  \hat P_{\phi_{L}(2)} $ is a product state then 
$\nu (\rho_{0}(1,2) ; \delta)$ has reduced conditions $\tilde W(1) = \nu(\hat P_{\phi_{R}(1)} ;\delta )$ and $\tilde W(2) = \nu(\hat P_{\phi_{L}(2)} ;\delta )$.\end{proposition}

For an entangled two particle wave-function $\Psi(1,2) = (\alpha \phi^{+}(1)\otimes\phi^{+}(2) + \beta \phi^{-}(1) \otimes\phi^{-}2 )$ with orthogonal single particle wave functions, $\{ \phi^{+}_{k},  \phi^{-}_{k}\}\;; k=1,2$, the  entangled pure state is $\hat P_{\Psi(1,2)}$ and its reduced states are mixed states,  $\rho(k) = ( |\alpha|^{2} \hat P_{\phi^{+}(k)}  + |\beta|^{2}  \hat P_{\phi^{-}(k)} )$ for $k =1,2$. 

\begin{proposition}\label{PR244} If $ \rho_{0}(1,2) = \hat P_{\Psi(1,2)}$ is an entangled state then the condition $\nu (\rho_{0}(1,2) ; \delta)$ has reduced conditions  $\tilde W(k) =   |\alpha|^{2}\nu(\hat P_{\phi^{+}(k)},\delta) +  |\beta|^{2}\nu(\hat P_{\phi^{-}(k)},\delta)$ for $k=1,2$.\end{proposition}

\subsubsection{ When systems interact }
For a wide class of interactions in finite dimensional Hilbert spaces, Durt \cite{8}, has shown that quantum states become entangled. There is a similar result for the conditions of two particle systems that holds on Hilbert spaces of arbitrary dimensions. 

\begin{definition} \label{D5}An interaction is separable if its potential function satisfies \begin{equation}V(\vec q(1),\vec p(1), \vec q(2),\vec p(2)) = V_{1}(\vec q(1),\vec p(1)) + V_{2}(\vec q(2),\vec p(2)).\end{equation} \end{definition} A classical  example is the small oscillations of a spherical pendulum, for which the potential energy is $V(q(1), q(2)) = \frac{1}{2}(q(1)^{2} + q(2)^{2})$. It provides independent equations of motion for the variables $q(1)$ and $q(2)$. A non-separable interaction would produce coupled equations.

\begin{theorem}\label{TH6}For a two particle system the joint condition becomes entangled when the particles interact via a non-separable interaction.
\end{theorem}\begin{proof} Using Hamiltonian equations, see $\S$ \ref{equn}, it is clear that if the particles were prepared  in a product condition $ W_{0} = \nu(\hat P_{\phi(1)\otimes \phi(2)};\delta) = \nu( \hat P_{\phi(1)}, \frac{\delta}{2}) \otimes \nu( \hat P_{\phi(2)}, \frac{\delta}{2})$, with unit vectors $\phi(j) \in \mathcal{H}(j),\; j=1,2$ and $0 <\delta <\frac{1}{2}$, then under a separable potential $W(t)$ stays a product condition.
 
When the particles interact via a non-separable potential, the equations of motion for the individual particles are coupled so that after the interaction has ceased $(q_{1}\otimes q_{2})|_{W(t)}  \neq q_{1}|_{W(t)} q_{2}|_{W(t)}.$ \end{proof}

For a one dimensional example take an impulsive von Neumann interaction. 
\begin{lemma}\label{L6} Let $H_{I}( q_{1}|_{W}, p_{1}|_{W},  q_{2}|_{W}, p_{2}|_{W}) = \gamma q_{1}|_{W}p_{2}|_{W} $, and $\kappa = \gamma T$, where $T$ is the duration of the impulse. Then \begin{equation}
q_{2}|_{W(T)} = q_{2}|_{W(0)} + \kappa q_{1}|_{W(0)}, \;\;\; q_{1}|_{W(T)} = q_{1}|_{W(0)}\end{equation} so that  \begin{equation}(q_{1}\otimes q_{2})|_{W_{0}(T)}  \neq q_{1}|_{\tilde W_{1}(T))} q_{2}|_{\tilde W_{2}(T)}\end{equation} $\tilde W_{1}(T)$, $\tilde W_{2}(T)$ are reduced conditions for particles $1$ and $2$ at time $T$.\end{lemma}
\begin{proof} Since Hamilton's equations are linear, particles $1$ and $2$ keep their trajectories whether we use the qr-number equations or Heisenberg's equations for the operators averaged over open sets of states.
 
From Heisenberg's operator equations, $\hat Q_{1}(T)\otimes \hat I_{2} = \hat Q_{1}(0)\otimes \hat I_{2} $ and $\hat I_{1} \otimes \hat Q_{2}(T) =\hat I_{1} \otimes \hat Q_{2}(0) + \kappa \hat Q_{1}(0)\otimes \hat I_{2}$.  Therefore $(\hat Q_{1}(0)\otimes \hat I_{2})(\hat I_{1}\otimes \hat Q_{2}(T)) =  \hat Q_{1}(0)\otimes \hat Q_{2}(0) + \kappa \hat Q_{1}(0)^{2}\otimes \hat I_{2}$. SInce  $W_{0}$ is a product condition, every $\rho(1,2) \in W_{0}$ is a product state, so that $Tr \rho(1,2) \hat Q_{1}(T)\otimes \hat Q_{2}(T) = (Tr \rho_{1}\hat Q_{1}(0))(Tr \rho_{2}\hat Q_{2}(0)) + \kappa Tr \rho_{1} \hat Q_{1}(0)^{2}.$ 

Thus $(q_{1}\otimes q_{2})|_{W(T)} = q_{1}|_{\tilde W_{1}(0)} q_{2}|_{\tilde W_{2}(0)} + \kappa q_{1}|_{\tilde W_{1}(0)}^{2} \neq q_{1}|_{\tilde W_{1}(0)} q_{2}|_{\tilde W_{2}(0)}$. Therefore the joint condition condition $W_{T}$ has become entangled. 
\end{proof}

In \S III of Corbett and Home's paper \cite{11}, the preparation of a two particle entangled state is described using an impulsive von Neumann interaction, $\hat H = \gamma \hat Q_{S}\cdot \hat P_{M}$, and time-dependent coordinate wave functions. Under disjointness assumptions  on the supports of the functions $\psi^{+}_{S} = \psi_{+}(q_{S}, t_{1})$ and $\psi^{-}_{S} = \psi_{-}(q_{S}, t_{1})$ and assuming that $\phi_{M} = \phi_{0}(q_{M}, t_{1})$ is an approximate eigenfunction\footnote{ For the meaning of approximate eigenvector/value see Weyl's criterion in Reed and Simon \cite{27}, pp237 and pp 364 for unbounded self-adjoint  operators} of position they obtain an entangled wave function $\Psi_{S,M}(t_{2}) = (a \psi^{+}_{S}\otimes \phi^{+}_{M} + b  \psi^{-}_{S}\otimes \phi^{-}_{M})$, with $|a|^{2} + |b|^{2} = 1$ and both $ \psi^{+}_{S} \perp  \psi^{-}_{S}$ and $ \phi^{+}_{M} \perp  \phi^{-}_{M}$. Although the coordinate spaces of $S$ and $M$ were assumed to be one dimensional in \cite{11}, the argument extends to 3 dimensional coordinate spaces. For $s = \pm$, the wave-functions $\phi^{s}_{M}$ are given by convolutions,  see \cite{10} \S 0.C, \begin{equation}\label{EQ19} \phi_{s}(q_{M}, t_{2}) = \int|\psi_{s}(q_{S},t_{1})|^{2}\phi_{0}(q_{M}-\Gamma(t_{2}) q_{S},t_{1})dq_{S},
\end{equation} where $\Gamma(t_{2}) = \gamma (t_{2} -t_{1}).$ 

The evolution of the wave function $\Psi_{S,M}(t_{1}) = (\alpha\psi^{+}_{S} + \beta \psi^{-}_{S}) \otimes \phi_{M} $ into an entangled wave function $\Psi_{S,M}(t_{2}) = \alpha\psi^{+}_{S}\otimes \phi^{+}_{M} + \beta \psi^{-}_{S} \otimes \phi^{-}_{M}$
leads to the following evolution of the conditions.
 
\begin{theorem}\label{TH7} Under the unitary group $\hat U(t) = \exp( \imath \hat H t/\hbar)$, for an impulsive interaction $\hat H = \gamma \hat Q_{S}\cdot \hat P_{M}$, the condition $\nu( \hat P_{\psi_{S,M}(t_{1})}  ,\delta)$ evolves to $\nu( \hat P _{\Psi_{S,M}(t_{2})}, \delta)$ with  $\Psi_{S,M}(t_{2}) = (\alpha \psi^{+}_{S}\otimes \phi^{+}_{M} + \beta  \psi^{-}_{S}\otimes \phi^{-}_{M})$.\end{theorem}

Since the wave function  $\Psi_{S,M}(t_{1})$ evolves to the wave function $\Psi_{S,M}(t_{2})$ then the state $\hat P_{\Psi_{S,M}(t_{1})}$ evolves to the state $\hat P_{\Psi_{S,M}(t_{2})}$ and, by Lemma \ref{L2}, the condition $\nu(\hat P_{\Psi_{S,M}(t_{1})}, \delta) $ evolves to $\nu(\hat P_{\Psi_{S,M}(t_{2})}, \delta) $ and the condition $\mathcal{N}( \hat P_{\Psi_{S,M}(t_{1})}, \hat Q_{S}, \delta)$ evolves to $\mathcal{N}( \hat P_{\Psi_{S,M}(t_{2})}, \hat Q_{S}, \delta)$.

\subsection{ Decomposing conditions}\label{decomp}If the condition is the union of basic open sets, $W = \cup_{j=1}^{n} \nu(\rho_{j};\epsilon_{j})$ with $\epsilon_{j} >0$ and if each $\rho_{j} =  \lambda \rho ^{+}_{j} + (1 - \lambda)\rho^{-}_{j}$ for $0 < \lambda  < 1$, then $W$ also has a convex decomposition, $W = \lambda W^{+} + (1-\lambda)W^{-}$ where $W^{\pm}  = \cup_{j=1}^{n} \nu(\rho^{\pm}_{j};\epsilon_{j})$.

The proof of this follows from the lemma concerning the decomposition of the basic open sets $\nu(\rho, \delta)$ and the fact that every open set is a union of basic open sets.

\begin{lemma}\label{L9} If $\rho _{+}\neq \rho _{-}$ are distinct states in $\mathcal{E}_{S}(\mathcal{A})$ and $\rho_{0} = \lambda \rho _{+} + (1 - \lambda)\rho_{-}$ with $0 < \lambda <1$ then $\nu(\rho_{0}, \epsilon)$ can be decomposed following the decomposition of the state $\rho_{0}$;
$\nu(\rho_{0}, \epsilon) =  \lambda \nu( \rho _{+};\epsilon) +  (1 - \lambda)\nu( \rho _{-};\epsilon)$. 
\end{lemma}This true since $\nu  $ determines a norm $\|\cdot\|_{t}$ on the space of trace class operators, so that if $\sigma = \lambda \sigma_{+} + (1 - \lambda)\sigma_{-}$
with $\sigma_{s} \in \nu  (\rho_{s} ; \epsilon)$ for $s = \pm$, then $\sigma \in \nu(\rho_{0}, \delta)$ as
\begin{equation} 
\| \sigma - \rho_{0} \|_{t} \leq \lambda \| \rho_{+} - \sigma_{+} \|_{t} + (1 - \lambda)\| \rho_{-} - \sigma_{-} \|_{t} < \epsilon. 
\end{equation}
Conversely if $\rho_{m} \in \nu (\rho_{0} ;  \epsilon)$ then $ \rho_{ m} = \lambda \rho_{m}^{+} + (1 - \lambda)\rho_{ m}^{-} $ where $\rho_{ m}^{+} = \rho_{m} + (\rho_{+} - \rho_{0})$ and $\rho_{ m}^{-} = \rho_{ m} + (\rho_{-} - \rho_{0})$ hence $\rho_{ m}^{+} \in \nu( \rho_{+}; \epsilon)$ and  $\rho_{ m}^{-} \in \nu( \rho_{-}; \epsilon)$. Therefore $\nu (\rho_{0} ;  \epsilon) \subseteq \lambda \nu( \rho_{+};\epsilon) + (1 - \lambda) \nu(\rho_{-}; \epsilon)$.  These results are easily extended to finite convex sums.

Using a similar argument for the sub-basic open sets $\mathcal{N}(\rho_{0}, \hat A, \delta) = \{ \rho \in \mathcal{E}_{S}(\mathcal{A}) :  |Tr (\rho \hat A - \rho_{0} \hat A )| < \delta\}$,\begin{lemma}\label{L9} If $\rho _{+}\neq \rho _{-}$  and $\rho_{0} = \lambda \rho _{+} + (1 - \lambda)\rho_{-}$ with $0 < \lambda <1$ then $\mathcal{N}(\rho_{0}, \hat A, \delta) $ can be decomposed following the decomposition of $\rho_{0}$;
$\mathcal{N}(\rho_{0}, \hat A, \delta)  =  \lambda \mathcal{N}(\rho_{+}, \hat A, \delta)  +  (1 - \lambda)\mathcal{N}(\rho_{-}, \hat A, \delta) $. 
\end{lemma}

 Applying this to $\mathcal{N}(\rho_{0}, \hat A, \delta) $ when $\rho_{0} = \rho_{S}^{mix}= |\alpha|^{2} \hat P_{\psi_{S}^{+}} + |\beta|^{2} \hat P_{\psi_{S}^{-}}$, $\lambda = |\alpha|^{2}$, $(1- \lambda) = |\beta|^{2}$, $\rho_{+} = \hat P_{\psi_{S}^{+}}$, $\rho_{-} = \hat P_{\psi_{S}^{-}}$ and $\hat A = \hat Q_{S}$ then \begin{equation}
 \mathcal{N}(\rho_{S}^{mix},\hat Q_{S}, \delta) = |\alpha|^{2}\mathcal{N}(\hat P_{\psi^{+}_{S}}, \hat Q_{S}, \delta) + |\beta|^{2} \mathcal{N}(\hat P_{\psi^{-}_{S}}, \hat Q_{S}, \delta).\end{equation}
\begin{lemma}  If $\psi_{S} = \alpha \psi_{S}^{+} + \beta \psi_{S}^{-},$  
$\psi_{S}^{\pm}$ are orthonormal eigenvectors of $\hat Q_{S}$, $|\alpha|^{2} + |\beta|^{2} = 1$ and $\delta > 0$ then\begin{equation}
 \mathcal{N}( \hat P_{\psi_{S}}, \hat Q_{S}, \delta) =  |\alpha|^{2}\mathcal{N}(\hat P_{\psi^{+}_{S}}, \hat Q_{S}, \delta) + |\beta|^{2} \mathcal{N}(\hat P_{\psi^{-}_{S}}, \hat Q_{S}, \delta)
\end{equation}\end{lemma}  
\begin{proof} Firstly $ \mathcal{N}(\rho_{S}^{mix},\hat Q_{S}, \delta) = |\alpha|^{2}\mathcal{N}(\hat P_{\psi^{+}_{S}}, \hat Q_{S}, \delta) + |\beta|^{2} \mathcal{N}(\hat P_{\psi^{-}_{S}}, \hat Q_{S}, \delta)$ was shown above  and $Tr\hat P_{\psi_{S}} \hat Q_{S} = |\alpha|^{2} Tr \hat P_{\psi^{+}_{S}}\hat Q_{S} +  |\beta|^{2} Tr \hat P_{\psi^{-}_{S}}\hat Q_{S}$. If $\rho = |\alpha|^{2} \rho_{+} +  |\beta|^{2} \rho_{-}$ with $\rho_{t}\in \mathcal{N}( \hat P_{\psi^{t}_{S}}, \hat Q_{S}, \delta)$ for $t =\pm$ then 
$|Tr(\rho\hat Q_{S} - \hat P_{\psi_{S}}\hat Q_{S}) | = | |\alpha|^{2}Tr (\rho_{+}\hat Q_{S} - \hat P_{\psi_{S}^{+}}\hat Q_{S} ) + |\beta|^{2}Tr(\rho_{-} \hat Q_{S} - \hat P_{\psi_{S}^{-}}\hat Q_{S}) | <\delta,$ showing that  $|\alpha|^{2}\mathcal{N}(\hat P_{\psi^{+}_{S}}, \hat Q_{S}, \delta) + |\beta|^{2} \mathcal{N}(\hat P_{\psi^{-}_{S}}, \hat Q_{S}, \delta)\subseteq \mathcal{N}( \hat P_{\psi_{S}}, \hat Q_{S}, \delta)$. 

On the other hand if $\rho_{c} \in \mathcal{N}( \hat P_{\psi_{S}}, \hat Q_{S}, \delta)$ then $\rho_{c} = |\alpha|^{2} \rho_{c}^{+} +  |\beta|^{2} \rho_{c}^{-}$ where $\rho_{c}^{+} = \rho_{c} +(\hat P_{\psi_{S}^{+}} - \hat P_{\psi_{S}} )$ and $\rho_{c}^{-} = \rho_{c} +(\hat P_{\psi_{S}^{-}} - \hat P_{\psi_{S}} )$  so that $\rho_{c}^{+}\in \mathcal{N}(\hat P_{\psi^{+}_{S}}, \hat Q_{S}, \delta)$ and $\rho_{c}^{-}\in \mathcal{N}(\hat P_{\psi^{-}_{S}}, \hat Q_{S}, \delta)$, therefore $\mathcal{N}( \hat P_{\psi_{S}}, \hat Q_{S}, \delta)\subseteq  |\alpha|^{2}\mathcal{N}(\hat P_{\psi^{+}_{S}}, \hat Q_{S}, \delta) + |\beta|^{2} \mathcal{N}(\hat P_{\psi^{-}_{S}}, \hat Q_{S}, \delta).$\end{proof}

% If you have acknowledgments, this puts in the proper section head
\section*{Acknowledgments}
I wish to thank Professor Dipankar Home for introducing me to the quantum mechanical measurement problem many years ago. Any misinterpretations of the problem are my own doing.

% Create the reference section using BibTeX:
%\bibliography{your bib file}

\end{document}